\documentclass[a4paper]{tlp2}
\pdfoutput=1

\usepackage{ifthen}
\newboolean{commentsaon}         %
\setboolean{commentsaon}{true}
 \setboolean{commentsaon}{false}  %

\newboolean{commentson} %
\setboolean{commentson}{true}

\newcommand{\comment}[1]
{\ifthenelse{\boolean{commentson}\AND\boolean{commentsaon}}
   {{\par\noindent\mbox{}{\small\blue[ *** #1 ]\par}\noindent\par}}{}}

\newcommand{\commenta}[1]
{\ifthenelse{\boolean{commentsaon}}
   {{\par\noindent\mbox{}{\small\color[rgb]{0, .4, 0}[ *** #1 ]\par}\noindent\par}}{}}

\usepackage{xspace}
\usepackage{latexsym}
\usepackage{amssymb}   %

\usepackage{color}
\newcommand\blue     {\color{blue}}

\usepackage{hyperref}
\usepackage[hyphenbreaks]{breakurl}

\newtheorem{theorem}{Theorem}
\newtheorem{corollary}[theorem]{Corollary}
\newtheorem{lemma}[theorem]{Lemma}
\newtheorem{proposition}[theorem]{Proposition}
\newtheorem{example}[theorem]{Example}

\newtheorem{remark}[theorem]{Remark}

\newcommand*{\seq}[2][n]  {{#2_{1}, \allowbreak \ldots, \allowbreak #2_{#1}}}
\newcommand*{\SEQ}[3]
            {{\ensuremath{#1_{#2}, \allowbreak \ldots, \allowbreak #1_{#3}}}}

\newcommand*{\notmodels}{\mathrel{\,\not\!\models}}

\newcommand*{\HU}{{\ensuremath{\cal H U}}\xspace}

\newcommand*{\M}{{\ensuremath{\cal M}}\xspace}
\usepackage[mathscr]{euscript} 
\newcommand*{\F}{{\ensuremath{\mathscr F}}\xspace}
\renewcommand*{\P}{{\ensuremath{\cal P}}\xspace}
\newcommand*{\X}{{\ensuremath{\cal X}}\xspace}

\usepackage[ddmmyyyy]{datetime}

\begin{document}

\title 
[Program answers and least Herbrand models]
%
%
%
{On definite program answers \\ and least Herbrand models%
}

\author[W. Drabent]
       {W{\l}odzimierz Drabent\\
         Institute of Computer Science,
         Polish Academy of Sciences,\\
         ul. Jana Kazimierza 5,
         01-248 Warszawa, Poland
         \\ and \\
         Department of Computer and Information Science,
         Link\"oping University\\
         S -- 581\,83   Link\"oping, Sweden      \\
         \email{drabent\,{\it at}\/\,ipipan\,{\it dot}\/\,waw\,{\it dot}\/\,pl}
}

\submitted
{11-01-2016}

%

%
%

\maketitle

\begin{abstract}
A sufficient and necessary condition 
is given under which least Herbrand models exactly characterize the answers
of definite clause programs. 

To appear in Theory and Practice of Logic Programming (TPLP)
\end{abstract}


\section{Introduction}
The relation between answers of definite logic programs and their least Herbrand
models is not trivial.
In some cases the equivalence
\begin{equation}
\label{eq}
    \M_P\models Q \ \ \mbox{ iff } \ \ P\models Q  
\end{equation}
does not hold
(where $P$ is a definite program, $\M_P$ its least Herbrand model, and $Q$ a
query, i.e.\ a conjunction of atoms%
\footnote
{%
The semantics of non closed formulae is understood as usually
(see e.g.\  \cite{vanDalen,Apt-Prolog}),
so that ${\it IT}\models Q$ iff
${\it IT}\models\forall Q$, where ${\it IT}$ is an interpretation or a
theory,  $Q$ a formula, and $\forall Q$ its
universal closure.%
}%
).
So programs with the same least Herbrand model may have different sets of
answers. 
(By definition, $Q$ is an answer of $P$ iff $P\models Q$.)
For a simple counterexample \cite[Exercise 4.5]{Doets},
assume that the underlying language has only one function symbol, a constant $a$.
Take a program $P=\{\, p(a)\,\}$.  
Now $\M_P\models p(X)$ but $P\notmodels p(X)$.
This counterexample can be in a natural way generalized for any  finite set
of function symbols,
see the comment following the proof of Prop.\,\ref{prop:counterexample}.

Equivalence (\ref{eq}) holds for ground queries
(\citeNP{Lloyd87}, Th.\,6.6;\,\,\citeNP{Apt-Prolog}, Th.\,4.30).
For a possibly nonground $Q$ (and a finite $P$)
a sufficient condition for (\ref{eq})
is that there are infinitely many constants in the
underlying language
(\citeNP
{DBLP:books/mk/minker88/Maher88};  \citeNP[Corollary 4.39]{Apt-Prolog}).
\citeN{DBLP:books/mk/minker88/Maher88} states without proof
that instead of an infinite supply of constants it is sufficient that there is
a non constant function symbol not occurring in $P,Q$. 
The author is not aware of any proof of this property
(except for \cite[Appendix]{drabent.arxiv.coco14}).

This paper presents a more general sufficient condition,
and shows that the condition is also a necessary one.
To obtain the sufficient condition, we show a property of (possibly
nonground) atoms containing symbols not occurring in a program $P$. 
Namely, when such atom is true in $\M_P$ then, under certain conditions, a
certain more general atom is a logical consequence of $P$.
As an initial step, we obtain
a generalization of the theorem on constants \cite{shoenfield67},
for a restricted class of theories, namely definite clause programs.
We also give an alternative proof for the original theorem.

\paragraph{Related problem.}

This paper studies (in)equivalence of
two views at the declarative semantics of definite
clause programs.  One of them considers answers true in the least Herbrand
models of programs, the other -- answers that are logical consequences of
programs.

The subject of this paper should be compared with a related issue
(which is outside of the scope of this paper).
There exists (in)equivalence between the declarative
semantics and the operational one, given by SLD-resolution.  
  As possibly first pointed in
(\citeNP{DM87};\,\,\citeyearNP{DM88}),
two logically equivalent programs 
(i.e.\ with the same models, and thus the same logical consequences)
may have different sets of SLD-computed answers for the same query.
For instance take
$P_1= \{\, p(X).\,\}$, and 
$P_2= \{\, p(X).\ p(a).\,\}$
Then for a query $p(Y)$ program $P_2$ gives two distinct computed answers,
and $P_1$ one.  This phenomenon gave rise to the {\em s-semantics}, 
see e.g.\ \cite{DBLP:journals/tcs/Bossi09} for overview and references.

\paragraph{Preliminaries.}

We consider definite clause logic programs.  A query is a conjunction of atoms.
A query $Q$  is an {\em answer} (or a {\em correct answer}) of a program $P$  
iff $P\models Q$.
\citeN{Apt-Prolog} calls it a correct instance (of some query).
We do not need to refer to SLD-computed answers, as
 each computed answer is an
answer, and each answer is a computed answer for some query,
by soundness and completeness of SLD-resolution.
Similarly, we do not need to consider to which query $Q_0$ a given query is
an answer.

The Herbrand universe
(for the alphabet of function symbols of the underlying language) 
will be denoted by \HU, and the least Herbrand model of a program $P$ by $\M_P$.
Remember that $\M_P$ depends on the underlying language.
We require $\HU\neq\emptyset$.
Names of variables will begin with an upper-case letter.
Otherwise we use the standard definitions and notation of \cite{Apt-Prolog},
including the list notation of Prolog.
(However in discussing the semantics of first order formulae we use a
standard term ``variable assignment'' instead of ``state'' 
used in \cite{Apt-Prolog}.)

The paper is organized as follows.  The next section presents some 
necessary definitions.
Section \ref{sec:lemma} 
shows how existence of answers containing symbols not occurring in the program
implies existence of more general answers.  The main result of this section
is compared with theorem on constants \cite{shoenfield67}.
Section \ref{sec:main}
contains the central technical lemma of this paper.
Section \ref{sec:H}
studies when the least Herbrand models provide an exact characterization of
program answers.  A new sufficient condition for equivalence
(\ref{eq}) is presented, and it is shown in which sense the condition is a
necessary one.

\section{Definitions}

This section introduces three notions needed further on.
Let $\F$ be the set of function symbols of the underlying language;
let $F\subseteq \F$.  An {\em alien} w.r.t.\ $F$  
is a non-variable term with its main function symbol from $\F\setminus F$.
An alien w.r.t.\ a theory $T$ (for instance a program)
means an alien w.r.t.\ the set of function symbols occurring in $T$.
An occurrence of an alien $t$ (w.r.t.\ $F$, in an atom or substitution)
will be called a {\em maximal alien} if the occurrence is not within
an alien $t'\neq t$.

By a generalization of a query we mean the result of systematic replacement
of maximal aliens in the query by new variables.  More formally, 
let \P be a theory or a set of function symbols.
Let the maximal aliens of a query $Q$ w.r.t.\ $\P$
be the occurrences in $Q$ of distinct terms $\seq t$.
Let $\seq V$ be distinct variables not occurring in $Q$.
Let a query $Q'$ be
 obtained from $Q$ by replacing (each occurrence of) 
$t_i$ by $V_i$, for $i=1,\ldots,n$.  
(So $Q=Q'\{V_1/t_1,\ldots,V_n/t_n\}$.)
Such $Q'$ will be called
{\em $Q$ generalized} for $\P$.
We will also call it a/the {\em generalization} of $Q$ (for \P).
Note that it is unique up to variable renaming. 

\begin{example}
The standard append program APPEND \cite[p.\,127]{Apt-Prolog}
contains two function symbols  $[\,]$ and $[\ |\ ]$.
Terms $a,f([a,b])$ are aliens w.r.t.\ APPEND, term $[a,b]$ is not.
Maximal aliens in $A = app( [a], [[\,]\,|\,g(a,X)], [g(a,Y),Z,[a]] )$
are the first and the last occurrences of $a$ and the (single) occurrences of
 $g(a,X)$ and $g(a,Y)$.
Atom $app( [V_1], [[\,]|V_2], [V_3,Z,[V_1]] )$
is $A$ generalized for APPEND.
\end{example}

Let $Q'$ be a query not containing aliens w.r.t.\ \P, and 
$\theta$ be a substitution such that $Dom(\theta)\subseteq {\it Var}(Q')$.
Then $Q'$ is a generalization of  $Q'\theta$ for \P (and for $\P\cup\{Q'\}$)
    iff  $\theta=\{V_1/t_1,\ldots,V_n/t_n\}$
    where $\seq t$ are distinct aliens w.r.t.\ \P.

 The correspondence between a ground atom and its generalization
 is described, in other terms, in
 \cite[Def.\,4]{naish.tplp.floundering14}.
 It is used in that paper to represent nonground atoms by
 ground ones, in analysis of floundering in the context of delays.

\section{On program answers and aliens}
\label{sec:lemma}

Given a query containing aliens which is an answer of a program $P$,
this section shows which more general queries are answers of $P$.
The main result (Lemma \ref{lemma:alien:substitution})
is compared with theorem on constants,
used by \cite{DBLP:books/mk/minker88/Maher88} to prove 
equivalence (\ref{eq}) for a case with an infinite alphabet of constants.

It is rather obvious that answers containing aliens can be generalized.
Assume that a query
 $Q$ is an answer of $P$, and that $Q$ contains aliens w.r.t.\ $P$. 
Then $Q$ is
a proper instance of some computed answer $Q'$.
It is however not obvious which replacements of aliens in $Q$ by variables 
result in answers.

\begin{example}
By replacing aliens w.r.t.\ $P$ by variables in an answer $Q$, 
we obtain some queries which are answers of $P$,
and some which are not.
Let $P= \{p(X,X,Y)\}$ and $Q=p( f(a),f(a),b )$.
So $P\models Q$.
Now $p(f(V_1),V_2,b)$ and $p(V_1,V_2,b)$ are not answers of $P$, 
but $p(f(V),f(V),Z)$,  $p(V,V,b)$ and  $p(V,V,Z)$ are.
\end{example}

\begin{lemma}
\label{lemma:alien:substitution}  
Let $P$ be a program, $Q$ a query, and
 $\rho=\{V_1/t_1,\ldots,V_k/t_k\}$ be a substitution where
 $\seq[k]t$ are distinct aliens w.r.t.\ $P\cup\{Q\}$.  Then
\begin{equation}
\label{th:alien:substitution1}
P\models Q \ \  \mbox{ iff } \ \  P\models Q\rho\,.
\end{equation}
\end{lemma}

Note that terms $\seq[k]t$ may be nonground 
(and may contain variables from $\{\seq[k]V\}$), some
$t_i,t_j$ may be unifiable, or contain common variables,
$Q$ may contain variables other than $\seq V$
and may contain aliens w.r.t. $P$.
So $Q$ is not necessarily a generalization of $Q\rho$ for $P$,
but it is one for $P\cup\{Q\}$.

\begin{example}
In the previous example, the cases in which the more general atom is an
answer of $P$ satisfy conditions of Lemma \ref{lemma:alien:substitution},
and the remaining ones do not.

\end{example}

\begin{proof}[Proof (Lemma \ref{lemma:alien:substitution})]

Without loss of generality assume that variables $\seq[k]V$ occur in $Q$.
Let  $\seq[l]X$ be the remaining variables of $Q$.
The ``only if'' case is obvious.

Assume $P\models Q\rho$.  
By completeness of SLD-resolution,
$Q\rho$ is an instance of some computed answer $Q\varphi$ for $P$ and $Q$:
$Q\rho=Q\varphi\sigma$.
Each function symbol occurring in $\varphi$ occurs in $P$ or $Q$.
Moreover (for $i=1,\ldots,k$) $t_i=V_i\varphi\sigma$
and the main symbol of $t_i$ does not occur in $V_i\varphi$;
hence $V_i\varphi$ is a variable.
As $\seq[k] t$ are distinct, variables $V_1\varphi,\ldots,V_k\varphi$ are
distinct. 
  Similarly,  $X_j=X_j\varphi\sigma$ for $j=1,\ldots,l$, thus
  $V_1\varphi,\ldots,V_k\varphi,X_1\varphi,\ldots,X_l\varphi$ 
  are distinct variables.
  Thus $Q\varphi$ is a variant of $Q$ and, by soundness of
  SLD-resolution, $P\models Q$.
\end{proof}

\begin{corollary}
\label{th:generalized}
\label{cor:generalized}
Let $P$ be a program, $Q$ a query, and $Q'$ be $Q$ generalized for $P$.
Then $P\models Q$ iff  $P\models Q'$.
\end{corollary}

\begin{proof}
$Q=Q'\rho$ for a certain  $\rho=\{V_1/t_1,\ldots,V_k/t_k\}$.
The premises of Lemma \ref{lemma:alien:substitution}  
are satisfied by $P$, $Q'$, and $\rho$
(as $\seq[k]t$ are aliens w.r.t.\ $P$, but also w.r.t.\ $Q'$).
\end{proof}

\begin{example}
Consider again program APPEND.  Assume that the underlying language 
has more function symbols than those occurring in the program, 
i.e.\  $[\,]$,  $[\ |\ ]$.
Assume that we know that the least Herbrand model $\M_{\rm APPEND}$
contains an atom
  $Q =\linebreak[3] app([\seq[m]t], [\SEQ{t}{m+1}k], [\seq[k]t])$,
where $\seq[k]t$ are distinct aliens w.r.t.\ APPEND.
Note that $P\models Q$, as equivalence (\ref{eq}) holds for ground queries.
{\sloppy\par}

By Corollary \ref{cor:generalized},
  ${\rm APPEND}\models  app([\seq[m]V], [\SEQ{V}{m+1}k], [\seq[k]V])$,
where $\seq[k] V$ are distinct variables.
Hence, for any terms $\seq[k] s$,
\mbox{%
${\rm APPEND}\models  app([\seq[m]s], [\SEQ{s}{m+1}k], [\seq[k] s])$%
}.
{\sloppy\par}
\end{example}

\begin{example}
 Consider the map colouring program  \cite[Program 14.4]{Sterling-Shapiro}.
 We skip any details, let us only mention that the names of colours and
 countries do not occur in the program.
 (The function symbols occurring in the program are 
 $F=\{\, [\,],
 [\, |\,],
 {\it region}   \,\}
 $.)
 By Corollary  \ref{cor:generalized},
 for any answer $Q$ of the program,
 the generalization $Q'$ of $Q$ w.r.t.\ $F$ is an
 answer of the program.  So is each instance of $Q'$.  
 Thus systematic replacing (some) names of colours or countries in $Q$ by
 other terms results in a query $Q''$ which is an answer of the program.%
\footnote{%
Thus it is possible that neighbouring countries get the same colour.
This does not mean that the program is incorrect.
Its main predicate {\it color\_map}
describes a correct map colouring provided that its second argument
is a list of distinct colours.
} 
\end{example}

The proof of equivalence (\ref{eq})
for an infinite set of constants of 
\cite[proof of Prop.\,6]{DBLP:books/mk/minker88/Maher88} employs a so called
theorem on constants  \cite{shoenfield67},
see also free constant theorem in \cite[p.\,56]{HandbookLAILP:FOL}.
The theorem states that (\ref{th:alien:substitution1}) holds for an arbitrary
theory $P$ and formula $Q$, when the distinct aliens $\seq[k]t$ are constants.
Its proofs in \cite{shoenfield67,HandbookLAILP:FOL} are syntactical, but a
rather simple semantic proof is possible:

Let $F$ be the set of function and predicate symbols from $P,Q$, let 
\X be the set of the free variables of $Q$.
Notice that
for any interpretation $I$ (for $F$) and any variable assignment $\sigma$
(for \X)
there exists a variable assignment $\sigma'$
(for $\X\setminus\{\seq[k]V\}$)
and an interpretation $I'$ 
(for $F\cup\{\seq[k]t\}$) such that
$\sigma'(X)=\sigma(X)$ for each $X\in\X\setminus\{\seq[k]V\}$,
 $I'(t_i)=\sigma(V_i)$ for each $i$, and all the symbols
of $F$ have the same interpretation in $I$ and $I'$.  
Thus $I\models P$ iff $I'\models P$, and 
$I\models_\sigma Q$ iff $I'\models_{\sigma'} Q\rho$.
Conversely, for each interpretation $I'$ for $F\cup\{\seq[k]t\}$
and variable assignment $\sigma'$ for $\X\setminus\{\seq[k]V\}$
there exist $I,\sigma$ as above. (In particular, the two equivalences hold.)
Now the theorem follows:
\smallskip

\begin{tabular}{l}
$P\models Q$ iff  \\
for every $I,\sigma$ (as above) $I\models P$ implies $I\models_{\sigma}Q$
iff \\
for every $I',\sigma'$ (as above) $I'\models P$ implies
$I'\models_{\sigma'}Q\rho$ iff \\
$P\models Q\rho$.
\end{tabular}
\smallskip

\citeN[p.\,634]{DBLP:books/mk/minker88/Maher88}
states that ``The same effect 
[as adding new constants] could be obtained with one new function symbol (of
arity $>0$) to obtain new ground terms with new outermost function symbol.''
This idea does not apply to the proof of the previous paragraph;
when $\seq[k]t$ are such terms then the proof fails.%
\footnote{%
  Informally, this is because such new terms cannot be interpreted
  independently, in contrast to $k$ new constants.
  Sometimes no interpretation for the new symbol $f$ is possible,
  such that  $\seq[k]t$ are interpreted as a given $k$ values.
  For instance take $t_i=f^i(a)$ for $i=1,\ldots,k$.
  Then for any interpretation for $f$,
  if $t_1,t_2$ have the same value then all $\seq[k]t$ also have the same value.
} %
So do the proofs of  \cite{shoenfield67,HandbookLAILP:FOL}.
In the context of \cite{shoenfield67} -- first order logic with equality -- the
generalization of the theorem on constants to terms with a
new outermost symbol does not hold.
For a counterexample, note that 
$\{a=b\}\models f(a)=f(b)$ \ but \ 
$\{a=b\}\notmodels V_1=V_2$.
The generalization in Lemma \ref{lemma:alien:substitution} is
sound and has a simple proof, due to restriction to 
definite programs and queries.

From Lemma \ref{lemma:alien:substitution} it follows that equivalence (\ref{eq})
holds whenever the underlying language has a non constant function symbol $f$
(or a sufficient number of constants) not occurring in $P,Q$.%
\footnote{%
    Assume that  $\seq[k]V$ are the variables of $Q$, and that there exist
    distinct ground terms $\seq[k]t$ with their main symbols
    not occurring in $P,Q$.
    Let  $\rho=\{V_1/t_1,\ldots,V_k/t_k\}$.
    Assume $\M_P\models Q$, so $\M_P\models Q\rho$, 
    and $P\models Q\rho$ as $Q\rho$ is ground.
    By Lemma \ref{lemma:alien:substitution}, $P\models Q$. 
} %
(See also \cite[Appendix]{drabent.arxiv.coco14} for a direct proof.)
We however aim for a more general sufficient condition for (\ref{eq}),
allowing $f$ to occur in $Q$;
in this case Lemma \ref{lemma:alien:substitution} is not applicable.

\section{Least Herbrand models and program answers}
\label{sec:main}

This section shows conditions under which
truth in $\M_P$ of a query with aliens implies 
that a certain more general query is an answer of $P$.
This is a central technical result of this paper (Lemma \ref{lemma:MP}).  
From it,
the sufficient conditions for equivalence (\ref{eq}) follow rather
straightforwardly, as shown in the next section.
We begin with proving an auxiliary property, by means the two following lemmas.

\begin{lemma}
\label{lemma:unifier}
Two distinct terms have at most one unifier of the form $\{X/u\}$ where $u$
is not a variable.  
\end{lemma}

\begin{proof}

Let $\theta=\{X/u\}$,  $\theta'=\{X'/u'\}$ be distinct substitutions,
where neither of $u,u'$ is a variable.
We show that if $s_1\theta=s_2\theta$ then  $s_1\theta'\neq s_2\theta'$, 
for any distinct terms $s_1,s_2$.
The proof is
by induction on the sum $|s_1|+|s_2|$ of the sizes of $s_1,s_2$.
(Any notion of term size would do, providing that $|t|<|t'|$ whenever $t$ is
a proper subterm of $t'$.)
Assume that the
property holds for each  $s_1',s_2'$ such that $|s_1'|+|s_2'|<|s_1|+|s_2|$.

Let $s_1\neq s_2$ and $s_1\theta=s_2\theta$.
Notice that at most one of  $s_1,s_2$ is a variable.
(Otherwise $s_1\theta,s_2\theta$ are $s_1,s_2$ -- two distinct variables,
 or exactly one of  $s_1\theta,s_2\theta$ is a variable, contradiction.)
Assume that exactly one of $s_1,s_2$, say $s_1$, is a variable.
Then $s_1=X$ (as $s_1\theta\neq s_1$), so
$X$~does not occur in $s_2$ (as $X,s_2$ are unifiable),
hence $s_2\theta=s_2=u$.
Now if $X'\neq X$ then $s_1\theta'=X$ which is distinct from any instance of
$s_2$.  Otherwise  $X'= X$, hence $s_1\theta'=u'\neq u=s_2=s_2\theta'$.

If both $s_1,s_2$ are not variables then
$s_i=f(\seq[l]{{s_i}})$, for $i=1,2$.  For some~$j$, ${s_1}_j\neq{s_2}_j$
and $|{s_1}_j|+|{s_2}_j|<|{s_1}|+|{s_2}|$.
By the inductive assumption, 
${s_1}_j\theta'\neq{s_2}_j\theta'$;
thus ${s_1}\theta'\neq{s_2}\theta'$.
\end{proof}

\begin{lemma}
\label{lemma:distinct}
let \P be a theory or a set of function symbols.
Let $\seq[m]t$ be a sequence of distinct terms, where $\seq t$ 
($0\leq n \leq m$) are variables, and  $\SEQ t{n+1}m$ are aliens w.r.t.~$\P$.
Assume that if  $\SEQ t{n+1}m$ are ground then there exist ground
aliens $\seq u$ w.r.t.\ $\P$, pairwise distinct from  $\SEQ t{n+1}m$.
Then the sequence has
a ground instance $(\seq[m]t)\sigma$ consisting of $m$ distinct aliens
w.r.t.\ $\P$. 
\end{lemma}
\nopagebreak

\begin{proof}
Consider first the case of $\SEQ t{n+1}m$ ground.
Then $\sigma=\{t_1/u_1,\ldots,\linebreak[3]t_n/u_n\}$
is a substitution providing the required instance. 
{\sloppy\par}

Let some $t_j$ ($n<j\leq m$) be nonground. 
Its main symbol, say $f$, is a non-constant function symbol not
occurring in $\P$. Thus the set $Al$ of ground aliens w.r.t.\ $\P$ is infinite.
\pagebreak[3]

Let $\seq[l]X$ be the variables occurring in  $\seq[m]t$.
For some $s_1\in Al$ substitution $\theta_1=\{X_1/s_1\}$ is not a unifier of any
pair $t_i,t_j$ ($1\leq i<j\leq m$),
as by  Lemma \ref{lemma:unifier} 
each such pair has at most one unifier of the form $\{X_1/s\}$, $s\in\HU$.
Thus $(\seq[m] t)\theta_1$ is a sequence of $m$ distinct terms.
Applying this step repetitively we obtain the required sequence
$(\seq[m] t)\theta_1\cdots\theta_l$ of distinct ground terms.
\end{proof}

\begin{lemma}
\label{lemma:MP}
Let $P$ be a program,
$Q$ an atom, and  $Q'$ be $Q$ generalized for $P$.
If
\begin{enumerate}[(a)]
\item 
\label{lemma:MP:condition1}
the underlying language has a non-constant function symbol not occurring in~P,
or 

\item 
\label{lemma:MP:condition2}
$Q$ contains exactly $n\geq0$ (distinct) variables, and
the underlying language has (at least) $n$ constants not occurring in $P,Q$,
\nopagebreak
\end{enumerate}
\nopagebreak
then $\M_P\models Q$ iff $P\models Q'$.
\end{lemma}

\begin{proof}
Note that $Q= Q'\varphi$ where 
$\varphi=\{\, X_1/u_1,\ldots X_m/u_m\,\}$,  $\seq[m]X$ are the variables 
of $Q'$ not occurring in $Q$, and $\seq[m]u$ are the maximal aliens in $Q$ 
(precisely: the distinct terms whose occurrences in $Q$ are the maximal aliens
w.r.t.\ $P$).
Let $\seq[n] Y$ be the variables occurring in $Q$.

We construct a ground instance $Q\sigma$ of $Q$, such that 
$Q'$ is $Q\sigma$ generalized for $P$.
To apply Lemma \ref{lemma:distinct} to terms $\seq Y,\seq[m]u$,
note that if $\seq[m]u$ are ground then there exist $n$ ground aliens
w.r.t.\ $P$ pairwise distinct from $\seq[m]u$.
(They are either the constants from condition \ref{lemma:MP:condition2}, 
or can be taken from the infinite set of ground aliens w.r.t.\ $P$ with the main
 symbol from condition  \ref{lemma:MP:condition1}.)
By Lemma \ref{lemma:distinct}, 
there exists a ground instance  $(\seq Y,\seq[m]u)\sigma$, 
consisting of $n+m$ distinct aliens w.r.t.\ $P$,
where the domain of $\sigma$ is $\{\seq Y\}$.

Note that
$\varphi\sigma=\sigma \cup \{\, X_1/u_1\sigma,\ldots, X_m/u_m\sigma\,\} $.
The substitution
maps variables $\seq Y,\seq[m]X$ to distinct aliens 
$(\seq Y,\seq[m]u)\sigma$ w.r.t.\ $P$.
So $Q'$ is $Q'\varphi\sigma$ generalized for $P$.
Thus $P,\ Q'\varphi\sigma$ and $Q'$ satisfy the conditions of 
Corollary~\ref{cor:generalized}.

Now $\M_P\models Q$ implies $\M_P\models Q\sigma$ and then
$P\models Q\sigma$ 
(as equivalence (\ref{eq}) from Introduction holds for ground queries).
As $Q\sigma=Q'\varphi\sigma$, by
Corollary \ref{cor:generalized} $P\models Q'$.
 The ``if'' case is obvious, as $Q$ is an instance of $Q'$.
\end{proof}

\begin{remark}
  The premises of Lemma \ref{lemma:MP} can be weakened by
stating that $Q,Q'$ are atoms such that $Q=Q'\varphi$ for a substitution
$\varphi=\{\, X_1/u_1,\ldots X_m/u_m\,\}$,
where $\seq[m]u$ are distinct aliens w.r.t.\ $P\cup\{Q'\}$, and variables
$\seq[m]X$ do not occur in $Q$.
\end{remark}

\begin{proof}
Obtained by minor modifications of the proof above.
The first sentence, describing $\varphi$, is to be dropped.
Each ``w.r.t.\ $P$'' is to be changed to ``w.r.t.\ $P\cup\{Q'\}$''$\!$.
In the third paragraph, substitution $\varphi\sigma$ together with $P$ and
$Q'$ satisfy the condition of Lemma \ref{lemma:alien:substitution}.
At the end of the proof, 
Lemma \ref{lemma:alien:substitution} should be applied instead of 
Corollary \ref{cor:generalized}.
\end{proof}

It remains to generalize Lemma \ref{lemma:MP} to arbitrary queries.

\begin{corollary}
\label{cor:MP}
Lemma \ref{lemma:MP} also holds for non-atomic queries.  Moreover, condition
\ref{lemma:MP:condition2} of the lemma can be replaced by:
\begin{enumerate}[(a)]
\setcounter{enumi}2
\item
\label{cor:MP:condition3}
for each atom $A$ of $Q$ with  $k\geq0$ (distinct) variables,
the underlying language has (at least) $k$ constants not occurring in $P,A$.
\end{enumerate}
\end{corollary}

\begin{proof}
Note that condition \ref{lemma:MP:condition2} implies condition
\ref{cor:MP:condition3}.
So assume that the latter holds.
Let  $Q=\seq[l]A$ generalized for $P$ be $Q'=\seq[l]{A'}$.
Then each $A_i'$ is $A_i$  generalized for $P$.
So Lemma \ref{lemma:MP} applies to each $A_i,A_i'$.  Thus $\M_P\models Q$ implies
$P\models A_i'$, for each $i=1,\ldots,l$.
Hence $P\models Q'$.
\end{proof}

\section{Characterization of program answers by the least Herbrand model}
\label{sec:H}

This section studies when the least Herbrand models exactly characterize the
program answers. 
First a sufficient condition is presented for equivalence (\ref{eq}) from
Introduction.
Then we show that the sufficient condition is also necessary.
Conditions \ref{th:MP:condition1}, \ref{th:MP:condition2} below are the same
as conditions \ref{lemma:MP:condition1}, \ref{cor:MP:condition3}
 of Lemma \ref{lemma:MP} and  Corollary \ref{cor:MP}.

\begin{theorem}
[Characterizing answers by $\M_P$]
\label{th:MP}
Let $P$ be a program, and $Q$ a query such that
\begin{enumerate}[(a)]
\item 
\label{th:MP:condition1}
the underlying language has a non-constant function symbol not occurring in~P,
or 
\item 
\label{th:MP:condition2}
for each atom $A$ of $Q$ with  $k\geq0$ (distinct) variables,
the underlying language has (at least) $k$ constants not occurring in $P,A$.
\nopagebreak
\end{enumerate}
\nopagebreak
Then $\M_P\models Q$ iff $P\models Q$.
\end{theorem}

Note that condition \ref{th:MP:condition1}
 implies that the equivalence holds for every query $Q$,
including queries containing the new symbol.
Also, it holds for every query $Q$ and every finite program $P$ when the
alphabet contains infinitely many function symbols, as then 
condition  \ref{th:MP:condition1} or \ref{th:MP:condition2} is satisfied.
From the theorem the known sufficient conditions follow: 
the alphabet containing infinitely many constants (and $P$ finite),
or $Q$ ground.

Condition \ref{th:MP:condition2}
is implied by its simpler version:
the language has $k\geq0$ constants not occurring in $P,Q$, 
 and each atom of $Q$ contains no more than $k$ variables.

\begin{proof}[Proof of Th.\,\ref{th:MP}]
Let $Q'$ be $Q$ generalized for $P$.
By Corollary \ref{cor:MP},
 $\M_P\models Q$ implies $P\models Q'$, hence $P\models Q$, as $Q$ is an
instance of $Q'$.
The reverse implication is obvious.
\end{proof}

We conclude with showing in which sense
the sufficient condition of Th.\,\ref{th:MP}
is also necessary.  As expected, it is strictly speaking not a necessary
condition for (\ref{eq}), 
as it is violated for some $P,Q$ for which (\ref{eq}) holds. 

\begin{example}

\noindent
Consider program APPEND and assume that the only function symbols of the
underlying language are $[\,]$, 
 $[\  |\ ]$.
 Let $Q=app( [X], [Y], [X,Y] )$.
Then $\M_{\rm APPEND}\models Q$ and ${\rm APPEND}\models Q$, 
but the condition of Th.\,\ref{th:MP} is violated.

On the other hand, consider a program $P$ of three clauses 
$  app(\,[\,],L,L\,)$.\,;
{\small
$app(\,[[\,]|K],L,[[\,]|M]\,) \gets \linebreak[3] app(\,K,L,M\,)$.\,;
\hspace{0pt plus .5ex}%
$app(\,[[H|T]|K],L,[[H|T]|M]\,) \gets \linebreak[3]  app(\,K,L,M\,)$.
}
Programs \mbox{APPEND} and $P$ have the same least Herbrand model but
different sets of answers,
as e.g.\ $P\notmodels Q$.
The condition of Th.\,\ref{th:MP} is violated by $P,Q$, and the equivalence
does not hold.
Note that $P$ cannot be used to append lists when new function symbols are
added to the language; $app([a],[b],[a,b])$ is then not an answer of~$P$.
\sloppy
\end{example}

 Roughly speaking, 
the sufficient conditions of Th.\,\ref{th:MP} and Lemma \ref{lemma:MP}
are also necessary, 
when all what is known about a program is the set of function
symbols employed in it:

\begin{proposition}
\label{prop:counterexample}
Let $\F$ be the set of function symbols of the underlying language, and
$F_0\subseteq \F$ be its finite subset.
Let $Q$ be a query, 
such that the predicate symbols of the atoms of $Q$ are distinct.
Assume that  $\M_P\models Q$ iff $P\models Q$, for each finite program $P$
such that 
$F_0$ is the set of function symbols occurring in $P$.  Then the 
sufficient condition of Th.\,\ref{th:MP} holds.
\end{proposition}
The proposition also holds when $F_0$ and the considered program $P$ are
infinite.

\begin{proof}
Let $Q$ be a query whose atoms have distinct predicate symbols.
Assume that the sufficient condition of Th.\,\ref{th:MP} does not hold.
We show that for a certain program $P$ 
(such that $F_0$ is the set of the function symbols occurring in $P$), 
\linebreak[3]
\mbox{$\M_P\models Q$} but $P\notmodels Q$.

As condition \ref{th:MP:condition1} of Th.\,\ref{th:MP} does not hold, 
all the non-constant function symbols of $\F$ are in $F_0$.
As condition \ref{th:MP:condition2} does not hold, 
there is an atom $A$ in $Q$ with $k$ distinct variables $\seq[k]Y$,
for which the number of constants from $\F\setminus F_0$ not occurring in
$A$ is $l<k$\/; let  $\seq[l]a$  be the constants.
The atom can be represented as $A = B[\seq[n]b,\seq[k]Y]$,
where
$\seq[n]b$ are those (distinct) constants of $A$ which are not in $F_0$.%
\footnote{%
 Formally, $B[\seq[n+k]t]$ can be defined as the instance
 $B\{V_1/t_1,\ldots V_{n+k}/t_{n+k}\}$ of an atom $B$, whose
 (distinct) variables are $\seq[n+k]V$, and 
 whose function symbols are from $F_0$.
}
So $\F\setminus F_0 = \{ \seq[l]a,\seq[n]b \}$.
Let $P_0$ be the set of atoms of $Q$ except for $A$.
Let ${\cal V} = \{\seq[n+k-1]X\}$ be $n+k-1$ distinct variables.
Let $P$ consist of the unary clauses of $P_0$ and the
unary clauses of the form $B[\seq[n+k]t]$ where 
(i)~$\{\seq[n+k]t\} =\cal V $
(so a variable occurs twice),
or (ii)~$\{\seq[n+k]t\} ={\cal V} \cup \{f(\vec Z)\} $
where $f\in F_0$, its arity is $m\geq0$,
and $\vec Z$ is a
tuple of $m$ distinct variables pairwise distinct from those in $\cal V$.
Note that $P$ is finite iff $F_0$ is.

Each ground atom $B'=B[\seq[n+k]u]$ (where $\seq[n+k]u\in\HU$)
is an instance of some clause of $P$,
as if $\seq[n+k]u$ are distinct terms then the main symbol of some of
them is in $F_0$, and $B'$ is an instance of a clause of the form (ii),
otherwise $B$ is an instance of a clause of the form (i).
Thus  $\M_P\models A$, hence $\M_P\models Q$
(as $P_0\models A'$ for each atom $A'$ of $Q$ distinct from $A$).
To show that  $P\notmodels Q$, 
add new constants $a_{l+1},\ldots,a_k$ to the alphabet.
Then $B[\seq[n]b,\seq[k]a]$ is not an instance of any clause of $P$, 
so $B[\seq[n]b,\seq[k]a]$ is false in the least Herbrand model of $P$
with the extended alphabet.
\end{proof}

The proof provides a family of counterexamples for a claim that 
$\M_P\models Q$ and $P\models Q$ are equivalent.
In particular, setting $Q=p(V)$ ($k=1$, $n=0$) results in 
$P = \{\, p(f(\vec Z)) \mid f\in F \,\}$,
a generalization 
of the counterexample from Introduction to any underlying
finite set $F$ of function symbols.

From the proposition it follows that a more general sufficient condition
(than that of Th.\,\ref{th:MP}) for
the equivalence of $\M_P\models Q$ and $P\models Q$ is impossible,
unless it uses more information about $P$ than just the set of involved symbols.

\section{Conclusion}

In some cases the least Herbrand model does not characterize the set of
answers of a definite program.  
This paper generalizes the sufficient condition for
$\M_P\models Q$ iff $P\models Q$, to
``a non-constant function symbol not in $P$, or $k$ constants not in $P,A$
for each atom $A$ of $Q$''.
It also shows
that the sufficient condition cannot be improved unless more is
known about the program than just which function symbols occur in it.
  As a side effect, it is shown
  which more general queries are implied to be answers of $P$ by $Q$ being an
  answer.

\paragraph{Acknowledgement.}
Comments of anonymous referees helped improving the presentation.

\bibliographystyle{acmtrans}
\bibliography{bibshorter,bibmagic,bibpearl}

\end{document}